\documentclass[a4paper]{llncs}

\usepackage{amstext}
\usepackage{amsmath}
\usepackage{amssymb}
\usepackage{url}
\usepackage{enumerate}
\usepackage{graphicx}
\usepackage{makecell}
\usepackage{listings}
\usepackage{wrapfig}
\usepackage{paralist}
\usepackage{xspace}
\usepackage{color}
\usepackage{times}
\usepackage{proof}
\usepackage{algorithm}
\usepackage[noend]{algpseudocode}
\usepackage{color}
\usepackage{capt-of}
\usepackage{longtable}

\usepackage{subfigure}
\title{Safe and efficient collision avoidance control for autonomous vehicles 
}

\author{
    Qiang Wang\inst{1,2} \and
    Dachuan Li\inst{2,3} \and
    Joseph Sifakis* \inst{2,4}
}

\institute{
SUSTech Academy for Advanced Interdisciplinary Studies, Shenzhen, China \and
School of Computer Science and Engineering, SUSTech , Shenzhen, China \and
SUSTech Intelligent transportation Center, Shenzhen, China \and
Verimag, Universit\'e Grenoble Alpes, France
}

\begin{document}

\maketitle

\begin{abstract}
We study a novel principle for safe and efficient collision avoidance that adopts a mathematically elegant and general framework abstracting as much as possible from the controlled vehicle's dynamics and of its environment.  Vehicle dynamics is characterized by pre-computed functions for accelerating and braking to a given speed. Environment is modeled by a function of time giving the free distance ahead of the controlled vehicle under the assumption that the obstacles are either fixed or are moving in the same direction. The main result is a control policy enforcing the vehicle’s speed so as to avoid collision and efficiently use the free distance ahead, provided some initial safety condition holds. 

The studied principle is applied to the design of  two discrete controllers, one synchronous and another asynchronous. We show that both controllers are safe by construction. Furthermore, we show that their efficiency strictly increases for decreasing granularity of discretization. We present implementations of the two controllers, their experimental evaluation in the Carla autonomous driving simulator and investigate various performance issues.
\end{abstract}

\begin{keywords}
Safe and efficient collision avoidance, Autonomous vehicles, Model based design
\end{keywords}

\section{Introduction}
\label{introduction}

As a fundamental requirement for autonomous vehicle control, the problem of collision avoidance has been widely investigated using a large variety of approaches and frameworks. 
 These involve control-based techniques,  game theory,  formal methods including reachability analysis and logic-based controller synthesis or the design of specific protocols. 
 Furthermore, the assumptions underlying the adopted frameworks vary regarding the level of modeling of the dynamics of the controlled vehicle, 
 the number of vehicles and the type of their trajectories or the nature of the controller stimuli.

In this problem, the key issue  is  the development of control algorithms of tractable complexity guaranteeing collision avoidance and making efficient use of the available space. 
 It should be emphasized that most of the existing results fail to satisfy at least one of these requirements. 
 Most results focus on performance optimization and only partially satisfy safety requirements. 
 Some results involve decision processes requiring computationally heavy analysis and others propose theoretically correct solutions that are not robust when discretized. 
 Finally, some results put emphasis exclusively on safety under various scenarios and neglect performance which is not acceptable for cars; 
 lack of performance can become a safety issue as for instance in an overtaking maneuver.

We study a novel principle for safe and efficient collision avoidance.
 We adopt a mathematically simple and general framework making abstraction of the controlled vehicle's specific dynamics and of its environment,
 and using only three functions:
 (1) the free distance function $F(t)$ which determines for the vehicle the estimated free distance from the closest obstacle ahead at time $t$;
 (2) the accelerating function $A(V, v)$ which gives the distance travelled by the vehicle when accelerating from initial speed $V$ to speed $v$;
 (3) the braking function $B(V, v)$ which gives the distance travelled by the vehicle when braking from $V$ to speed $v$ ($v < V$).

The principle consists in the application of a simple induction rule.
 If at some time $t$ the speed of the vehicle with respect to the distance $F(t)$ is safe, i.e. $B(V, 0) \leq F(t)$,
 then the speed will be controlled to remain safe under the assumption 
 that $F(t)$ does not change faster than the vehicle can brake.
 This assumption always holds when the obstacles ahead are fixed or move in the same direction as the controlled vehicle.
 Furthermore, if safety can be guaranteed for speed $V$ and $B(V, 0) \leq F(t)$
 then in order to efficiently use the available space $F(t) - B(V,0)$ we apply an A/B (Accelerating/Braking) policy:
 we accelerate to a certain speed $v > V$, such that after acceleration it is still possible to safely brake from $v$.
 So efficiency boils down to computing the maximum target speed $v$ such that $0 \leq F - (A(V, v)+ B(v, 0))$.
 This computation may be costly depending on the properties of the accelerating and braking distance functions.
 The control principle consists in the dynamic application of A/B policies for a set of possible speed levels between speed 0 and the limit speed of the vehicle.
 For each speed level, it uses precomputed conditions for safely switching to adjacent speed levels depending on the free distance ahead.

We provide two different controllers for safe and efficient collision avoidance.
 The first controller is synchronous driven by periodically sampled values of the free distance $F$.
 The second controller is asynchronous receiving sporadically  available values of $F$.
 We prove that both controllers are safe and efficient, where efficiency means that based on the most recent value of $F$, 
 getting closer to the obstacle ahead may jeopardize safety.
 We also present their implementations and experimental evaluations in Carla autonomous driving simulator 
 and investigate various performance issues.

Our approach is characterized by the following:
\begin{enumerate}
\item It makes abstraction of the vehicle dynamics through the use of accelerating and braking functions that provide all the information needed for safe and efficient control. These functions are a kind of  "contact" between the controller and the controlled vehicle. Their use frees us from the obligation to model vehicle dynamics. Furthermore, it leaves completely open the way features related to comfort such as the jerk profile are implemented.

\item  Although we consider the problem in one dimension and the environment is modeled by a free distance function $F(t)$, the result can be easily extended to two dimensions. 
In that case $F(t)$ and $B(v,V)$ become areas and the safety test consists in checking their inclusion.

\item The control principle is robust and easy to adapt to varying uncertainty in the measurement of F or in the estimation of the functions $A$ and $B$.

\item The proposed implementations do not have any specific hardware requirements and require very limited computing resources as they combine pre-computed control policies.

\item Finally, the adopted control principle is simple and inductive: if at some step the distance is safe then a speed increase by some quantity will not jeopardize safety.  This induction hypothesis is used to prove correctness.
\end{enumerate}

The paper is organized as follows.
 Section~\ref{relatework} reviews related work.
 Section~\ref{sec:collision} presents the framework and the principle of safe and efficient collision avoidance control.
 Section \ref{sec:design} presents the design of the collision avoidance controllers.
 Section~\ref{impl-expr} presents the implementations and performance evaluations using the Carla simulator.
 Section~\ref{conclusions} concludes about the relevance of the results and outlines directions for future work.

\section{Related work}
\label{relatework}

Collision avoidance has been extensively studied in the context of adaptive cruise control.
 Most work addresses the problem by applying optimization techniques \cite{zhang2017optimization,park2009obstacle,wang2019obstacle}.
 For instance,  \cite{zhang2017optimization} models the ego vehicle and the obstacles around as convex sets,
 and  generates collision-free trajectories by solving a set of smooth non-convex constraints.
 In \cite{park2009obstacle}, safe trajectories are calculated based on a non-linear model predictive control approach for both lateral and longitudinal movements.
 The work in \cite{wang2019obstacle} applies the concept of artificial potential field and identifies five stages in the process of obstacle avoidance.
 In \cite{provebecorrectacc}, a hierarchical framework consisting of a nominal controller and an emergency controller has been studied. 
 The former is based on model predictive control (MPC) strategy and operates under normal conditions 
 to achieve passenger comfort without considering safety guarantee, 
 while the latter takes over if the headway approaches clearance distance constraints and ensures provable safety. 
 However, the scheme considers multiple leading vehicles and designs a controller for each of them, thus incurring increased computational cost.
 Finally, several works deal with collision avoidance methods relying on rich environment information
 (e.g., position, speed, width of the surrounding vehicles) from V2I or V2V communication \cite{milanes2012fuzzy,li2014rear}.
 Despite the promising results achieved by such approaches,
 optimization-based and potential field-based collision avoidance strategies do not allow safety guarantees,
 which are essential for autonomous vehicles.
 Additionally, they may lead to high computation cost in real-world implementations.

To ensure guaranteed safety and achieve correctness-by-construction,
 Mobileye \cite{shalev2017formal} advocated the application of model-based techniques.
 The proposed concept of Responsibility Sensitive Safety (RSS) relies on the computation of the safe distance between vehicles.
 It is argued that if a vehicle maintains the required safe distance from other vehicles,
 it will never be responsible for an accident even if it might still become involved in an accident.
 Different estimates of safety distance are proposed under the assumption of constant response time for acceleration and deceleration.
 Furthermore, in order to avoid the unnecessarily large gap between vehicles caused by situation-unaware strategies in \cite{shalev2017formal},
 it is shown how to improve the safe distance conditions by taking into account the state of the ego-vehicle (pause, acceleration, deceleration) \cite{li2018situation}.
 Nonetheless, this work focuses only on conditions guaranteeing safety and does not address control issues.
Similarly, nVIDIA proposes a formal safety model, namely the Safety Force Field (SFF) \cite{nister2019safety},
 which brings in the concept of Safety Potential to evaluate the safety of traffic actors.
 An SFF function is defined to derive safety procedures that move an actor down the gradient of safety potential,
 resulting in actors repelling from each other when safety procedures are about to overlap.
 As RSS, SFF exclusively focuses on safety, and does not address efficiency issues.
 As a matter of fact, a study \cite{zhao2019rightofway} reveals that simply implementing RSS requirements leads to undesirable clearance distance and thus to decrease of traffic efficiency.

Model-based design for autonomous vehicles have been an active research area since 30 years.
 In the California PATH (Partners for Advanced Transportation and Highways) program,
 the concept of platoon has been proposed to mitigate the highway congestion.
 A platoon is a group of closely spaced vehicles under automatic control.
 In \cite{varaiya1993smart}, the design of platoon controllers has been investigated and
 a multi-layer automated highway system architecture has been proposed in order to achieve a fully automated platoon control.
 In \cite{lygeros1998verified},
 the analysis of \cite{varaiya1993smart} is refined using hybrid controllers and sufficient safety conditions are provided.
 Finally, \cite{horowitz2000control} presents the safety and performance analysis of a hybrid system modeling an autonomous vehicle.

There are several works involving application of formal methods to autonomous vehicle control.
 In \cite{stursberg2004verification}\cite{sadraddini2017provably}, 
 safety of the adaptive cruise control is verified by predicting and checking reachable states of ego and other vehicles, which is however computationally extensive.
 In \cite{ames2014control}, barrier certificates provide safety guarantee by defining a `barrier' that prevents transitions from safe states to unsafe ones.
 \cite{asplund2012formal} studies the distributed coordination of autonomous vehicles in order to avoid collisions in intersections.
 The coordination protocol is modeled in a constraint specification language and the automated constraint solver (i.e. Z3) is used to verify safety.
 \cite{krook2019design} presents the design and formal verification of a supervisor switching
 between nominal planners and a safe stop routine if nominal operational conditions are violated.
 In \cite{korssen2017systematic}, a supervisor for an advanced driver assistance system is automatically synthesized from the specifications modeled by finite state machines.
 The correctness of the switching logic is also formally verified.
 In \cite{nilsson2015correct}, a controller is synthesized from linear temporal logic specifications for adaptive cruise control.
 In \cite{hilscher2011abstract} an approach is presented for proving collision freedom of multi-lane traffic with lane-change maneuver.
 The multi-lane motorway traffic is modeled as an abstract transition system and the collision freedom property is specified in spatial logic.
 Then the safety verification problem boils down to checking that  the abstract transition system satisfies spatial logic formulas.
 In \cite{esterle2019specifications},  linear temporal logic is used to formalize traffic rules for both overtaking and merging maneuvers.
 Furthermore, these rules are verified on the automata modeling the behavior of an autonomous vehicle.
 A motion planner modeled as a maneuver automaton is presented in \cite{atva18motionplanner}.
 For each state of this model a particular motion control primitive is applied.
 The desired plan is specified by a formula of linear temporal logic, and logical correctness is reduced to to checking satisfiability of the formula.
 Finally, several papers discuss the application of formal verification to the decision and control software of autonomous vehicles e.g., \cite{loos2011adaptive,zita2017application}.

\section{Safe and efficient collision avoidance control}
\label{sec:collision}

The aim is to control the movement of a vehicle travelling in a one-way lane,
 so as to 1) avoid collision with other objects that may be fixed or moving in the same direction (i.e., safety);
 and 2) use the available free distance ahead in the best possible manner to minimize travelling time (i.e., efficiency).

Our work relies on a mathematically abstract framework characterized by three functions.
 We denote by $v$ the speed variable of the vehicle and by $V$ its initial speed.

\begin{itemize}
  \item The function $F(t)$ gives the free distance at time $t$ between the controlled vehicle and the closest obstacle ahead, which either moves in the same direction or is stopped.
  \item  The braking function $B(V, v)$ is a partial function defined in the interval $0 \leq v \leq V$.
   It gives the distance travelled by the controlled vehicle when braking from the initial speed $V$ to a target speed $v$.
   In Fig.\ref{functions} it is graphically illustrated by the green curves.
   When the target speed $v = 0$ (i.e, the vehicle brakes to stop), this function is abbreviated as $B(V)$ for simplicity.
  \item  The accelerating function $A(V, v)$ is a partial function defined in the interval $V \leq v \leq V_L$,
  where $V_L$  is a given  limit speed for each vehicle.
  It gives the distance travelled by the vehicle when accelerating from an initial speed $V$ to a target speed $v$.
  In Fig.\ref{functions} it is graphically illustrated by the black curves.
\end{itemize}

\begin{figure}[h]
\centering
\includegraphics[width=0.8\textwidth]{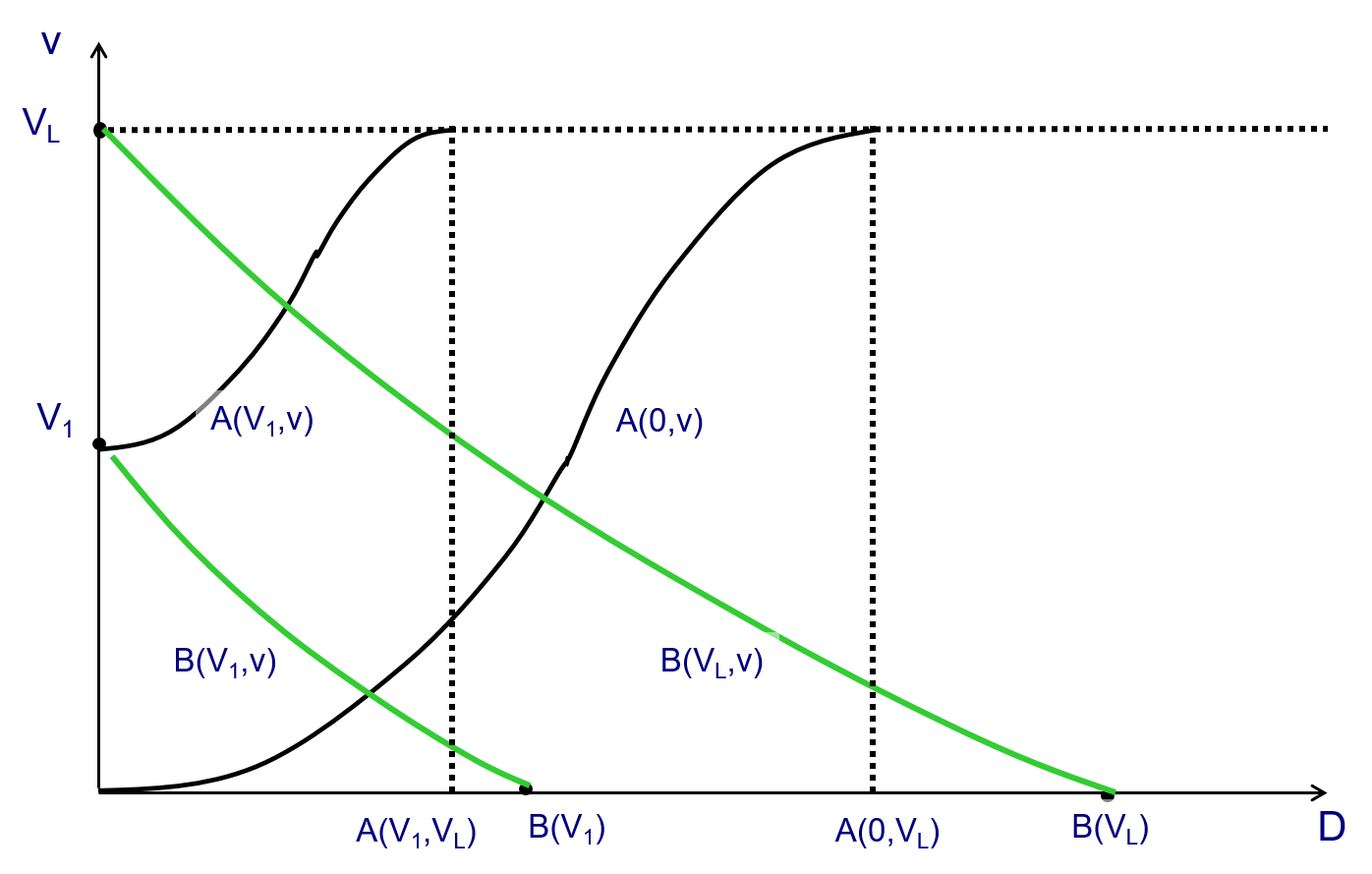}
\caption{Braking and acceleration distance functions ($D$ is the distance travelled and $v$ is the speed)}
\label{functions}
\end{figure}

We make no specific assumptions about the implementation of accelerating and braking functions, 
 e.g. whether acceleration and deceleration are constant or variable.
 Nonetheless, we require that the following properties hold:

\begin{itemize}
  \item  $B(V,V) = 0$ and $A(V,V) = 0$.
  \item  Additivity property:
  \begin{equation*}
  B(V, v_1) + B(V_1, v_2) = B(V, v_2), ~~ where ~~ v_1 = V_1
  \end{equation*}
  \begin{equation*}
  A(V, v_1) + A(V_1, v_2) = A(V, v_2), ~~ where ~~ v_1 = V_1
  \end{equation*}
  \item  Strict monotonicity:
  \begin{equation*}
 B(V, v_1) < B(V, v_2), ~~ when ~~ v_1 < v_2
  \end{equation*}
  \begin{equation*}
 A(V, v_1) < A(V, v_2), ~~ when ~~ v_1 < v_2
  \end{equation*}
\end{itemize}

The additivity property implies that for $0 \leq j < i \leq n$,
 \begin{equation*}
 B(v_i, v_j) =  \sum_{k=0}^{i-j-1} B(v_{i-k}, v_{i-k-1}) 
 \end{equation*}
  \begin{equation*}
 A(v_j, v_i) = \sum_{k=0}^{i-j-1} A(v_{j+k}, v_{j+k+1})
  \end{equation*} 
 This says that the distance needed to brake or accelerate to a given speed is the same 
 no matter how braking and acceleration commands have been applied.

As an example, when acceleration and deceleration rates are positive constant, respectively $a>0$ and $b>0$, these functions are given by the following formulas:
 \begin{equation*}
 B(V, v) = (V^2 - v^2)/ (2*b)
 \end{equation*}
 \begin{equation*}
 A(V, v) = (v^2  - V^2) / (2*a)
 \end{equation*}

We progressively study the safe and efficient collision avoidance problem for a vehicle moving in a one-way lane.
 We first study the problem for a stationary obstacle ahead.
 Then we study algorithms that solve the problem for dynamically changing free distance.
 We assume that the movement is controlled using commands for accelerating and braking from a speed $V$ to some target speed $v$
 whose effect is modeled by the functions $A(V,v)$ and $B(V,v)$, respectively.

\subsection{Control for safety}

The basic idea for avoiding collision is to moderate the speed of the vehicle and
 anticipate the changes of the free space ahead so as to have enough distance and time to adjust and brake.
 If the vehicle moves with speed $V$ at time $t$,
 then for safety the free space ahead $F(t)$ should be longer than the braking distance $B(V)$,
 which is the minimal safe braking distance for speed $V$.
 The Theorem below formalizes this idea.

\begin{theorem} \label{theorem1}
If at time $t$ the speed $V_t$ of the vehicle is safe with respect to $F(t)$, i.e., $B(V_t) \leq F(t)$
 and for any time $t + \triangle t$ it is possible to set the speed to a value $V_{t + \triangle t}$
 such that the condition $F(t) - F(t+\triangle t) \leq B(V_t) - B(V_{t+\triangle t})$ holds,
 then the vehicle is always safe.
\end{theorem}

\begin{proof}
The condition $F(t) - F(t+\triangle t) \leq B(V_t) - B(V_{t + \triangle t})$ relates changes of $F(t)$ to the changes of speed $V$.
 It simply says that the free space ahead does not change faster than the distance that the vehicle travels in some interval $\triangle t$.
 It can be deduced from the safety assumption $0 \leq F(t)- B(V_t)$ and from the condition that $0 \leq F(t) - B(V_t) \leq F(t + \triangle t)- B(V_{t + \triangle t})$.
\end{proof}

\begin{figure}[h]
\centering
\includegraphics[width=\textwidth]{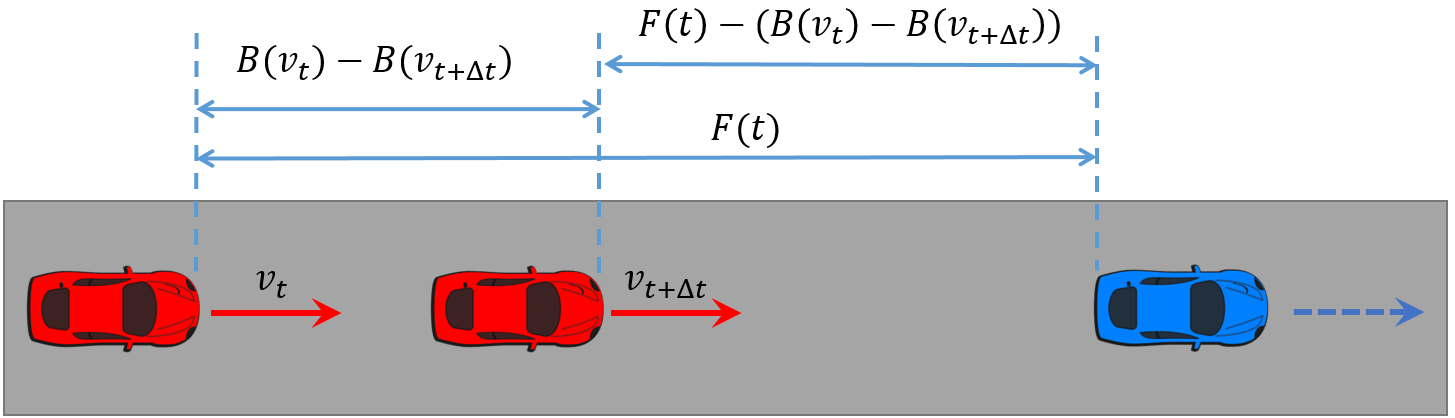}
\caption{Illustration of the control for safety}
\label{fig:safety}
\end{figure}

Notice as an application of the above theorem, that if the vehicle brakes from speed $V_t$ and the obstacles ahead do not move in the opposite direction,
 then the condition $F(t) - F(t+\triangle t) \leq B(V_t) - B(V_{t+\triangle t})$ trivially holds.
 In fact, when the vehicle brakes from $V_t$ for time $\triangle t$,
 it will reach the speed $V_{t + \triangle t} < V_t$ and it will have traveled the distance 
 $B(V_t, V_{t+\triangle t}) = B(V_t) - B(V_{t + \triangle t})$,
 by application of the additivity property.
 Then we have that $F(t) - (B(V_t) - B(V_{t + \triangle t}))$ is the distance ahead at time $t + \triangle t$ for the controlled vehicle,
 as shown in Fig.\ref{fig:safety}.
 By the assumption that the obstacles are moving forward or stopped, 
 we have that $F(t) - (B(V_t) - B(V_{t + \triangle t})) \leq F(t+\triangle t)$.
 Thus Theorem\ref{theorem1} can trivially be applied if obstacles ahead do not move in the opposite direction.

This theorem suggests a simple and safe control policy that ensures collision freedom.
 For any time $t$, the vehicle only needs to keep track of the free distance ahead $F(t)$ and
 check in real-time whether $F(t)$ is greater than the minimal safe braking distance $B(V_t)$ for the current speed $V_t$.
 It starts braking as soon as $F(t)$ reaches the minimal safe braking distance.
 In this way, it is guaranteed that if the obstacles ahead do not move in the opposite direction, no collision would happen.

\subsection{Achieving efficiency for fixed obstacles}

The above result provides a basis for ensuring collision freedom.
 Nonetheless, it leaves open the question of how the vehicle can efficiently use the available distance ahead by minimizing the travelling time.
 What would be an efficient driving policy when the free headway distance is greater than the minimal safe braking distance?
 We consider that a policy defines the speed function $v(t)$ in response to a free distance $F(t)$.
 An Accelerating/Braking policy (A/B policy) is a policy of accelerating first to some speed and then braking.
 Similarly, an Braking/Accelerating policy (B/A policy) is the policy of braking first to some speed and then accelerating.
 A Constant speed/Braking policy (C/B policy) is the policy of moving at constant speed and then braking.
 A policy  is safe if the relative distance between the controlled vehicle and the obstacle ahead is positive. 
 It is efficient if increasing the speed value $v(t)$ enforced by the policy at any point would compromise safety.

The problem is to minimize the travelling time for a given distance, which implies to maximize the average speed.
 Consider the scenario where the speed of the vehicle is $V$ and there is a stationary obstacle ahead at distance $F$, which is greater than the braking distance $B(V)$.
 The application of an A/B policy consists in computing an appropriate target speed $v, V < v \leq V_L$, accelerate the vehicle to $v$ and then brake to full stop.
 To ensure collision freedom, the total travelled distance $A(V, v) + B(v)$ must be such that $A(V, v) + B(v) \leq F$.
 The maximal target speed is given by the following condition.
\begin{equation*}
 v_M = \max \{v ~|~ F \geq A(V,  v) + B(v) \}
\end{equation*}
 Such a speed exists as both acceleration and braking functions are monotonically increasing with respect to the target speed $v$.
 Notice that either $v_M \leq V_L$ and $F = A(V, v_M) + B(v_M)$ or $v_M = V_L$ and $F > A(V, v_M) + B(v_M)$.

As an example, for motion at constant acceleration and deceleration ($a$ and $b$, respectively),
 we have $A(V, v) =  v * (v - V)/a + (v - V)^2/ (2*a)$ and $B(v) = v^2/(2*b)$.
 Then the safety condition becomes $F \geq v * (v - V)/a + (v - V)^2/(2*a) + v^2/(2*b)$, from which we deduce $v \leq \sqrt{(2*a*b*F + b* V^{2})/(a+b)}$.
 Thus the maximal target speed $v_M = \sqrt{(2*a*b*F + b* V^{2})/(a+b)}$.
 As we require that $v \geq V$, we have $F \geq V^{2}/(2*b) = B(V)$ and thus the maximal target speed always exists.
 Let $v_F$ denote the speed reached by accelerating along distance $F(t)$, i.e., $v_{F}^{2} - V^{2} = 2*F(t)*a$,
 then the formula can be simplified as $v_M = v_{F}*\sqrt{b/(a+b)}$.

\begin{figure}[h]
\centering
\includegraphics[width=0.8\textwidth]{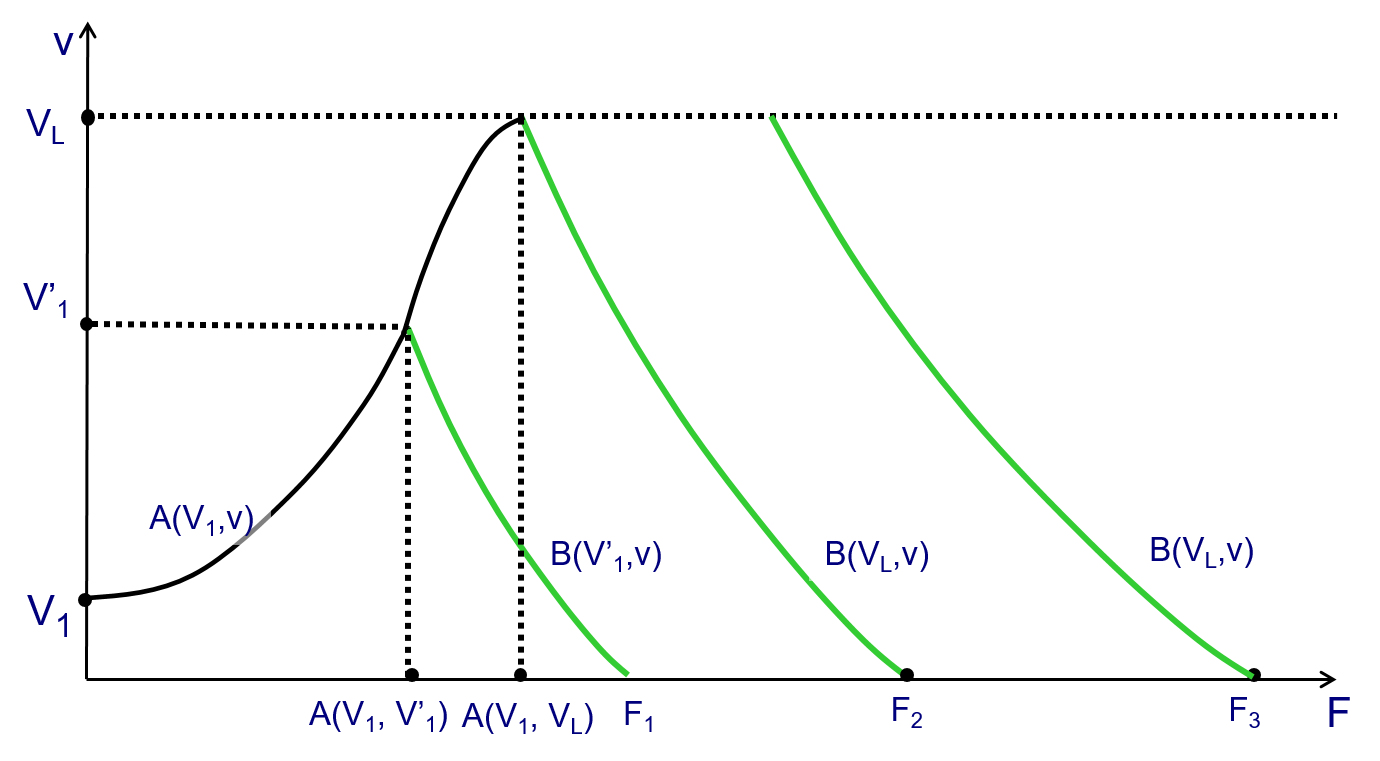}
\caption{The A/B control policy for different values of the free distance $F$ ahead}
\label{optimalspeed}
\end{figure}

Fig. \ref{optimalspeed} illustrates the A/B control policy where $F$ is the free distance ahead and $v$ is the speed of the controlled vehicle.
 The green curves illustrate braking phases and the black the accelerating phase from an initial speed $V_{1}$.
 For $F = F_1$, the maximal target speed $V'_{1}$ is less than the  limit speed $V_{L}$.
 The A/B policy consists in accelerating to $V'_{1}$, and then braking until the vehicle stops having travelled exactly distance $F_{1}$.
 If the free distance ahead is $F = F_{2}$, the maximal target speed will be the  limit speed $V_{L}$.
 The A/B policy will similarly accelerate first to the  limit speed and then brake to stop at $F_{2}$.
 Finally, if $F = F_{3} > F_{2}$, then after accelerating to the limit speed $V_{L}$,
 the vehicle will maintain constant speed $V_{L}$ for distance $F_{3} - A(V_1, V_L) - B(V_L)$ and then brake for the remaining distance to stop at $F_{3}$.

\begin{theorem} \label{theorem2}
If the speed $V$ of the vehicle is safe with respect to $F$, i.e.,$B(V) \leq F$, then the A/B policy is always safe and efficient for $F$.
\end{theorem}

\begin{proof}
The safety proof is given by the arguments following Theorem 1.
 To prove efficiency, we consider three basic driving policies: the A/B policy, the B/A policy and the C/B policy.
 The other possible policies, such as accelerating, driving at constant speed, accelerating and then braking, can be obtained as combinations of the three basic ones.
 We show that the A/B policy yields the minimal travelling time.

We decompose the free distance ahead $F$ into two segments: one segment of length $D = F - B(V)$ and one segment of length $B(V)$.
 Due to the additivity property, the distance $B(V)$ is always required regardless of the applied polices in order to brake safely from speed $V$.
 So the policies may differ only in the time needed to travel distance $D$.
 In the A/B policy, the vehicle travels distance $D$ by first accelerating to the maximal target speed $v_{M}$ and then braking from $v_{M}$ to $V$.
 We denote by $t_A$ the time needed to accelerate from $V$ to $v_{M}$ and by $t_B$ the time needed to brake from $v_{M}$ to $V$.
 In the C/B policy, the vehicle first moves at constant speed $V$ for the distance $D$ and then brakes from $V$ for the remaining distance $B(V)$.
 We denote by $t_{D}$ the time needed to travel $D$ with constant speed $V$.
 We show that $t_{D}$ is greater than $t_A + t_B$.
 We denote the speed function during acceleration by $v'(t)$ and the speed function during deceleration by $v''(t)$.
 Then we have $v'(t) > v$ for $t_A > t > 0$ and $v''(t) > v$ for $t_B > t > 0$.
 Since $D = v * t_{D} = \int_{0}^{t_A} v'(t) dt + \int_{0}^{t_B} v''(t)dt > \int_{0}^{t_A} v dt + \int_{0}^{t_B} v dt = v*(t_A + t_B)$, we have $t_{D} > t_A + t_B$.
 Thus the A/B policy takes less time and it is more efficient than the C/B policy.

In the B/A policy, the vehicle first brakes and accelerates for the distance $D$ and then brakes for the remaining distance $B(V)$.
 We denote by $t'_{B} + t'_{A}$ the travelling time of $D$ by braking and accelerating.
 By applying a similar reasoning, we can show that $t'_{B} + t_{A} > t_{D}$.
 Thus the C/B policy is more efficient than the B/A policy.
 This concludes the proof.
\end{proof}

The above result implies that for the given free distance $F$,
 the A/B policy is the most efficient and that from the given initial speed there is a maximal speed that minimizes the travel time of $F$.

\section{Controller design for collision avoidance}
\label{sec:design}

\subsection{The control principle}

We study a control principle for collision avoidance based on the above results.
 We consider that the vehicle speed can change between a finite set of increasing levels $v_0, v_1, ..., v_n$,
 where $n$ is a constant, $v_0 = 0$ and $v_n$ equals to the  limit speed $v_{L}$.
 The triggering of acceleration and braking from one level to another is controlled according to the free distance ahead and based on bounds computed as follows, for each speed level $v_i, i\in [1, n]$,

 \begin{itemize}
   \item $B_{i} = B(v_i)$ is the minimal safe braking distance needed for the vehicle to fully stop from speed $v_i$;
   \item $D_{i} = A(v_{i-1}, v_{i}) + B(v_i)$ is the minimal safe distance needed for the vehicle
 to apply an A/B policy accelerating from speed $v_{i-1}$ to $v_i$ and then braking from $v_i$ to stop.
 \end{itemize}

We show that the following function specifies the highest safe speed level $v$ as a function of the current speed of the vehicle $V$ and the free space ahead $F$, 
 provided that their initial values $V_0$ and $F_0$ are such that $B(V_0) \leq F_0$.
\begin{equation*}
  v = Control(F,V)
\end{equation*}
\begin{equation*}
 v =
    \begin{cases}
    v_{i+1} & \text{when ~~ $V = v_i \wedge F = D_{i+1}$} \\
    v_{i-1} &  \text{when ~~ $V = v_i \wedge F = B_{i}$} \\
    v_{i} & \text{when ~~ $V = v_i \wedge D_{i+1} > F > B_{i}$}
    \end{cases}
\end{equation*}

\begin{figure}[h]
\centering
\includegraphics[width=0.8\textwidth]{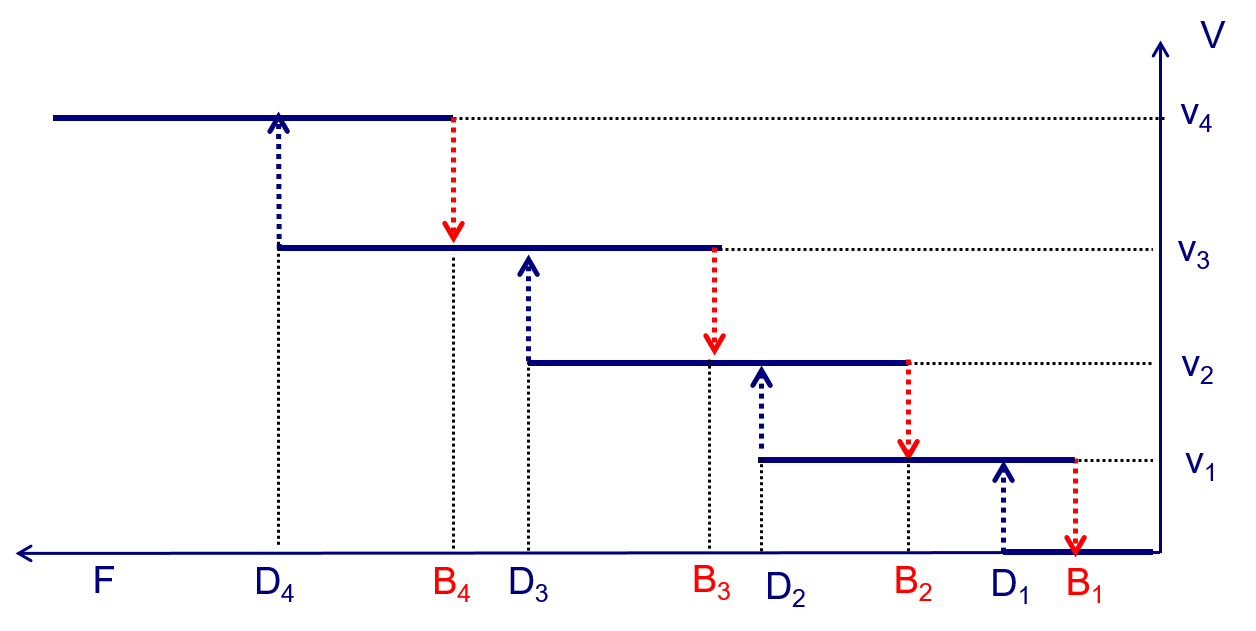}
\caption{Illustration of the collision avoidance principle for $n=4$}
\label{ledder}
\end{figure}

Note that this control principle is purely functional.
 It assumes that changes of the free distance ahead $F$ can be continuously monitored to instantaneously produce corresponding speed changes.

Fig.\ref{ledder} illustrates the principle for $n=4$ speed levels.
 As the value of $F$ increases, the speed of the vehicle switches between levels.
 Safety is preserved by construction.
 The vehicle can accelerate to a higher level, if it can safely and efficiently use the available distance by applying an A/B policy.
 It brakes to a lower level if the available distance reaches the bound for safe braking.

\begin{figure}
  \centering
  \includegraphics[width=0.9\textwidth]{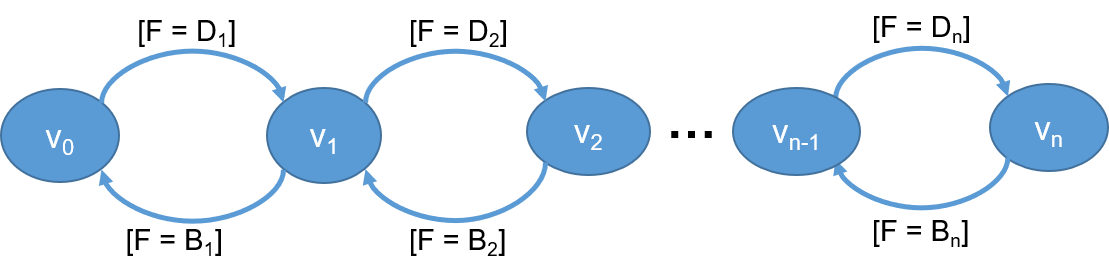}
  \caption{Automaton modelling the collision avoidance principle}
  \label{automaton}
\end{figure}

Fig.\ref{automaton} provides a scheme for the computation of $Control(F,V)$ in the form of a finite state automaton.
 The locations correspond to traveling at constant speeds $v_0, ..., v_n$.
 The transitions model instantaneous acceleration and braking steps 
 triggered by conditions involving the free distance $F$ and the precomputed bounds $B_{i}$ and $D_{i}$.
 If the control location is $v_{i}$ and the free distance ahead equals to the minimal safe acceleration distance (i.e., $F = D_{i+1}$),
 then the automaton moves to location $v_{i+1}$ after the speed is accelerated to $v_{i+1}$.
 If the free distance ahead reaches the minimal safe braking distance (i.e., $F = B_{i}$),
 then the automaton moves to location $v_{i-1}$ after the speed is decelerated to $v_{i-1}$.
 Recall that $B_{i} = B_{i-1} + B(v_i, v_{i-1})$.
 Thus, after braking to $v_{i-1}$ there is still enough space for safe braking.
 Note that checking point conditions makes sense because $F$ has no jumps and computation is instantaneous.
 If none of the triggering conditions holds, then the free distance ahead $F$ is such that $B_i < F < D_{i+1}$.
 The automaton stays at location $v_i$ and the speed remains unchanged.

Note that the automaton of Fig.\ref{automaton} cannot be implemented as a controller
 because we assume that $F$ is continuously observable and changes of the controlled speed are instantaneous.
 In the next section, we show how to design controllers by refining this automaton.

\begin{theorem}\label{theorem3}
The collision avoidance principle is safe. Moreover, its efficiency is strictly increasing for increasing number of speed levels $n$.
\end{theorem}

\begin{proof}
Safety can be proved by induction on the number of speed levels.
 First, we prove that the transition to  $ v_1$ is safe.
 If $v=0$ and the condition $F = D_1$ holds, speed can change to $v_1$.
 The transition to $v_1$  needs distance $A(v_0, v_1)$.
 When this speed is reached, the distance $F'$ will be such that $F' \geq (F - A(v_0, v_1)) = B(v_1)$, since $D_1 = A(v_0, v_1) + B(v_1)$.
 Thus the vehicle is still safe at $v_1$ because the remaining distance is greater than the minimal safe braking distance.

Assuming safety for $v = v_i$, we prove safety for $v = v_{i+1}$.
 Safety for $v = v_i$ means that $F \geq B_i$ and it will remain safe as long as the speed level does not change.
 We distinguish two cases.
 If $F \geq D_{i+1} = A(v_{i}, v_{i+1}) + B(v_{i+1})$, then we can accelerate to speed $v_{i+1}$
 and the free headway distance will be $F' \geq (F - A(v_{i}, v_{i+1})) = B(v_{i+1})$,
 which implies that the safety condition still holds for $v_{i+1}$.
 If $B_{i+1} < F < D_{i+1}$, then the speed remains unchanged and there is enough distance to brake for $v_{i+1}$.
 If $F = B_{i+1}$, then the vehicle will brake to the lower speed $v_{i}$, which is safe by assumption.
 Thus it is also safe for $v = v_{i+1}$.

Note that as speed can be enforced to discrete levels, efficiency is achieved only when $F = D_i$ for some $i$;
 otherwise, the highest safe speed level is chosen.
 Let $V_{init}$ be the speed of the vehicle
 and $F_{i}  = A(V_{init}, V_i)+ B(V_i)$ be the distance needed for the application of an A/B policy from  $V_{init}$ to $V_i$.
 Then we consider two cases:

\begin{itemize}
  \item Either $F = F_{d1}$ such that  $B(V_{init}) \leq F_{d1} \leq F_{n}$ in which case
 by construction there exists some $v_i$ such that  $v_{i} \leq V_{d1} \leq v_{i+1}$ as shown in Fig.\ref{optimalspeed},
 where $V_{d1}$ is the maximal speed such that $F_{d1} = A(V_{init}, V_{d1})+ B(V_{d1})$.
 In that case the controller will apply the best A/B policy to reach from $V_{init}$ the speed level $V_i$.
 The loss in efficiency $V_{d1} - V_i$ is determined by the max of the difference $V_{i+1} - V_{i}$ for $i \in [1, n]$.
 Thus, for increasing number of speed levels, the efficiency increases.
  \item  Or $F = F_{d2}$ such that $F_{n} < F_{d2}$.
 In that case the controller will accelerate to the allowed  limit speed $v_n$ and then will keep the speed constant for distance $F_C = F_{d2} - F_{n}$.
 Then if the obstacle ahead is fixed, it will have to brake for distance $B(v_n)$.
 In that case there is no loss of efficiency as the limitation comes from the limit speed of the vehicle.
\end{itemize}

 This concludes the proof.
\end{proof}

As explained, computing the exact value of the optimal speed for a given distance may be costly.
 Considering discrete speed levels allows pre-computing for each level both the minimal safe braking distance and the minimal safe accelerating distance between levels.
 In that manner, we avoid the computational complexity of adjusting in real time the vehicle speed.

\subsection{Controller design}

We propose two controllers applying the presented collision avoidance principle.
 The first controller is synchronous driven by periodic updates of the free distance variable $F$ for an adequately chosen period.
 The second controller is asynchronous in the sense that the free distance variable $F$ is updated sporadically.

\subsubsection{Synchronous controller}

The controller interacts with its controlled environment (the vehicle) through input and output events as shown in Fig.\ref{controller1-box}.
 The output $s$ is a state variable indicating the currently applied command (i.e., accelerating, braking or constant speed).
 The input event $UpdateF$ signals the periodic measurement $F'$ of the free distance with period $T$,
 while input events $ca$ and $cb$ signal the completion of the accelerating and braking command respectively.
 Initially, the speed $v$ of the vehicle is set to a level $v_i$ that is safe with respect to the initial distance $F$ (i.e., $F \geq B(v_i)$).

\begin{figure*}[!htb]
  \centering
  \includegraphics[width=\textwidth]{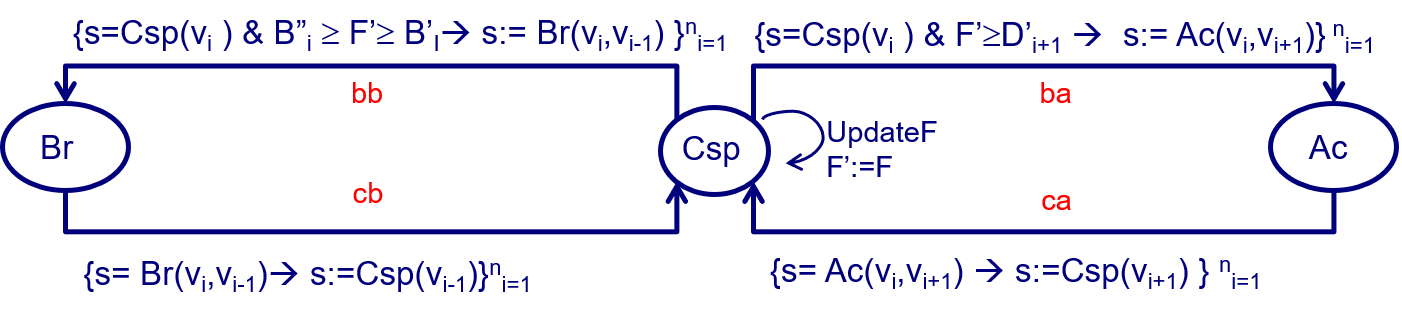}
  \caption{Extended automaton modelling the synchronous controller}
  \label{controller1}
\end{figure*}

 \begin{figure}[!htb]
  \centering
  \includegraphics[width=0.8\textwidth]{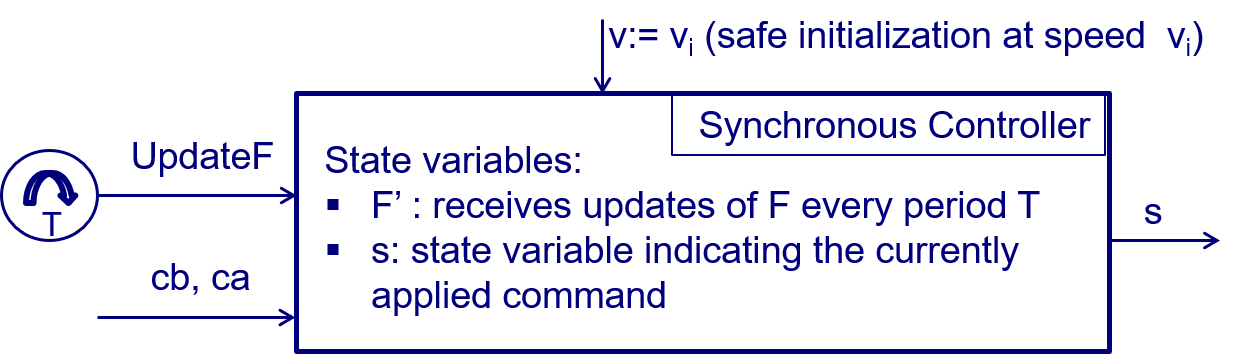}
  \caption{Inputs and outputs of the synchronous controller}
  \label{controller1-box}
\end{figure}

The controller is a refinement of the ideal controller where we assumed that speed changes were instantaneous. 
 It is described by the following set of guarded commands and also depicted as an extended automaton for the sake of clarity in Fig.\ref{controller1}.
\begin{equation*}
	\begin{split}
        do      & \\
         \Box ~  & \exists i \in [1,n]. s = Csp(v_{i}) \wedge  F' \geq D'_{i+1} \\
                 & ~~~~~~~         \rightarrow  s:= Ac(v_i, v_{i+1}); ~ba \\
         \Box ~  & \exists i \in [1,n]. s = Csp(v_{i}) \wedge  B''_{i} \geq F' \geq B'_{i} \\
                 & ~~~~~~~ \rightarrow s := Br(v_{i}, v_{i-1}); ~bb \\
         \Box ~  & \exists i \in [1,n]. ca \wedge s = Ac(v_i, v_{i+1}) \rightarrow  s:= Csp(v_{i+1}) \\
         \Box ~  & \exists i \in [1,n]. cb \wedge  s = Br(v_i, v_{i-1}) \rightarrow  s:= Csp(v_{i-1}) \\
         \Box ~  &  UpdateF  \rightarrow  F':= F \\
        od      & 
 	\end{split}
\end{equation*}
For guarded commands we adopt the usual semantics: 
 whenever the condition on the left hand side holds, the actions on the right hand side are executed. 
 Note that the input events appear as conditions while the output event appear as actions. 
 The variable $s$ keeps track of the kinematic state of the vehicle 
 that is abstracted by the control states \textbf{Ac} (accelerating), \textbf{Br} (braking) and \textbf{Csp} (moving with constant speed).

We denote by $Ac(v_i, v_{i+1})$, $Br(v_i, v_{i-1})$ and $Csp(v_i)$ the commands of accelerating from speed $v_i$ to $v_{i+1}$, 
 braking from $v_i$ to $v_{i-1}$ and moving with speed $v_i$, respectively.
 When the vehicle is moving with constant speed, transition \textbf{UpdateF} is triggered periodically to receive the most recent measurement of $F$.
 Once the triggering condition of accelerating (braking) is met, transition \textbf{ba} (\textbf{bb}) is taken to initiate the command 
 and move to location \textbf{Ac} (\textbf{Br}) waiting for its completion.

We do not make any assumption about the time spent at locations \textbf{Ac} and \textbf{Br}.
 We simply assume that the distances needed for accelerating and braking are $A(v_{i-1},v_{i})$ and $B(v_{i}, v_{i-1})$, respectively.
 We explain below how the guards of the controllable transtions $bb$ and $ba$ are computed.

We estimate the maximal safe approximations of the triggering conditions $F \geq D_i$ and $F = B_i$  of the ideal controller in terms of $F'$, the most recently updated value of $F$.
 When the vehicle moves at speed $v_i$,
 the variables $F$ and $F'$ satisfy a relation of the form $F = F' - k_{i}(t)$,
 where $k_{i}(t) = v_{i} * (t ~ mod ~ T)$.
 That is $k_{i}(t) = 0$ when $F$ is updated and $k_{i}(t) < v_{i} * T$.
 We assume that $T$ is small enough so that $D_{i} - v_{i} * T \geq B_{i}$, 
 that is, we do not miss the braking threshold value $B_{i}$ in a period.
 This is reasonable given that in practice the updating period $T$ is usually less than 50 milliseconds.
 Notice that the minimal value of $F$ will be reached for $F = F' - v_n * T$.
 Thus, it is enough to require that $F' \geq D_i + v_n * T$ holds for accelerating 
 and that $B_i + 2 * v_n * T \geq F' \geq B_i + v_n * T$ holds for braking.
 So we adjust the triggering bound for accelerating to $D_i' = D_i+ v_n * T$ and the least and upper bounds of the interval triggering a braking to $B_i' = B_i + v_n * T, B_i'' = B_i + 2* v_n * T$.

\subsubsection{Asynchronous controller}

Inputs and outputs of the asynchronous controller shown in Fig.\ref{controller2-box} differ in
 that the input event $UpdateF$ receiving the measurement $F'$ of the free distance occurs sporadically.
 Furthermore, this controller needs an internal clock event $tick$ with period $\triangle t$ to estimate the vehicle's position.

 \begin{figure}[!htb]
  \centering
  \includegraphics[width=0.8\textwidth]{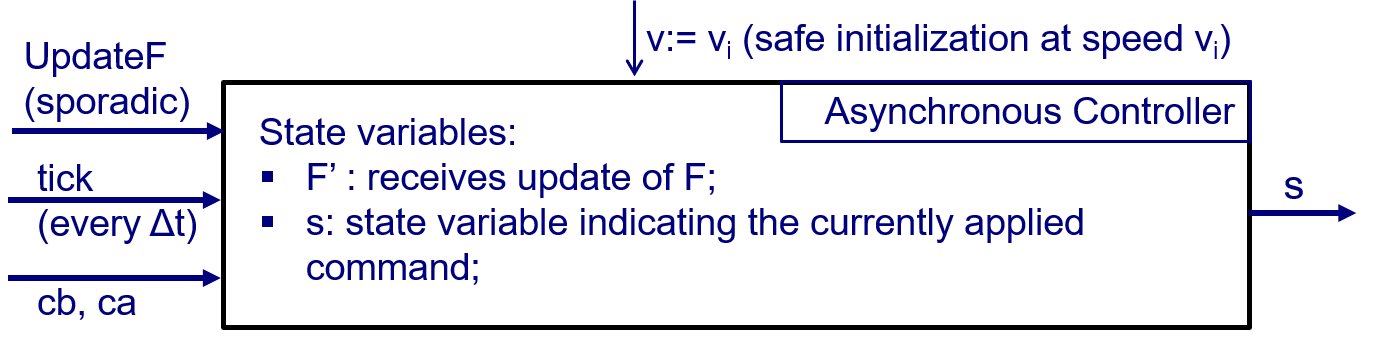}
  \caption{Inputs and outputs of the asynchronous controller}
  \label{controller2-box}
  \vspace{-0.2in}
\end{figure}

\begin{figure*}[!htb]
  \centering
  \includegraphics[width=\textwidth]{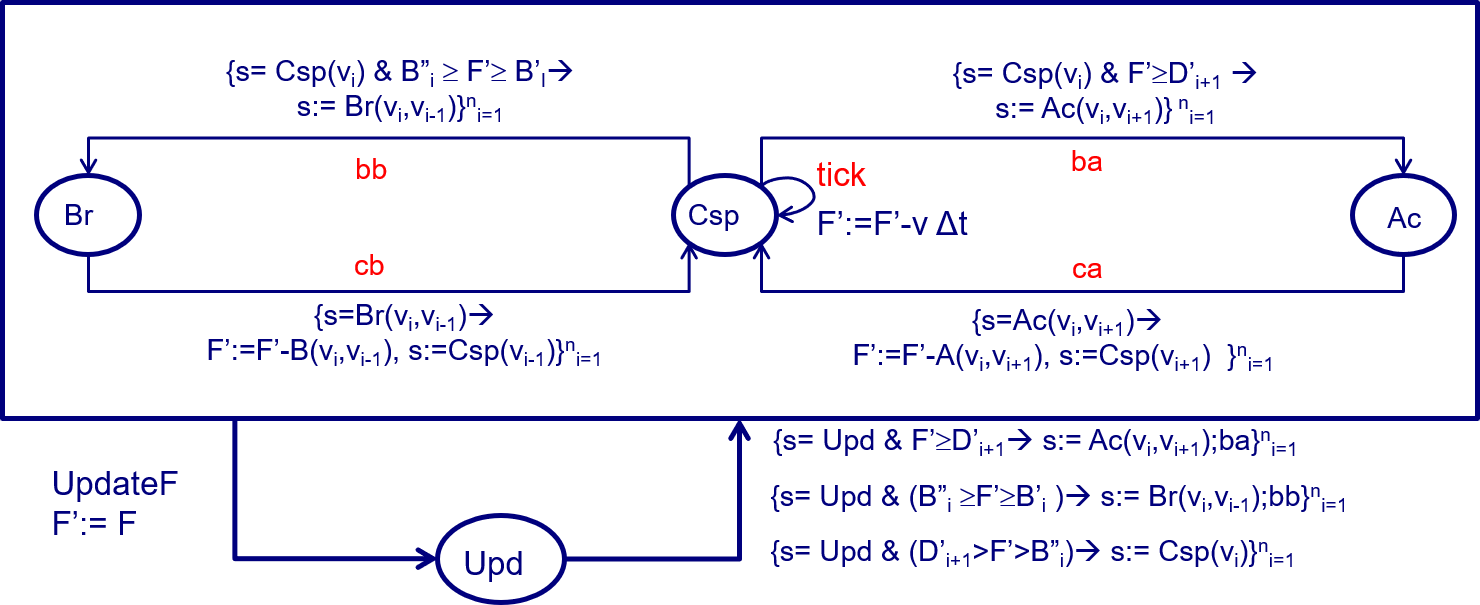}
  \caption{Extended automaton modelling the asynchronous controller}
  \label{controller2}
\end{figure*}

The asynchronous controller differs from the synchronous controller in the way of estimating the values of the free distance $F$.
 Its is described by the following guarded commands and also depicted as an extended automaton in Fig.\ref{controller2}.
\begin{equation*}
	\begin{split}
        do      & \\
         \Box ~  & \exists i \in [1,n]. (s = Csp(v_{i}) \vee s = Upd) \wedge  F' \geq D'_{i+1} \\
                 & ~~~~~ \rightarrow  s:= Ac(v_i, v_{i+1}); ~ba \\
         \Box ~  & \exists i \in [1,n]. (s = Csp(v_{i}) \vee s = Upd) \wedge  B''_{i} \geq F' \geq B'_{i} \\
                 & ~~~~~ \rightarrow s := Br(v_{i}, v_{i-1}); ~bb \\
         \Box ~  & \exists i \in [1,n]. s = Upd \wedge D'_{i+1} > F' > B''_{i} \rightarrow s := Csp(v_{i}) \\
         \Box ~  & \exists i \in [1,n]. ca \wedge  s = Ac(v_i, v_{i+1}) \\
                 & ~~~~~ \rightarrow  F':=F' - A(v_{i}, v_{i+1}) ; s:= Csp(v_{i+1}) \\
         \Box ~  & \exists i \in [1,n]. cb \wedge  s = Br(v_i, v_{i-1}) \\
                 & ~~~~~ \rightarrow  F':=F' - B(v_{i}, v_{i-1}) ; s:= Csp(v_{i-1}) \\
         \Box ~  & \exists i \in [1,n]. tick \rightarrow F' := F' - v_{i} * \triangle t \\
         \Box ~  & UpdateF \rightarrow  F':= F; s := Upd  \\
        od      & 
 	\end{split}
\end{equation*}
As previously we adopt similar notations. 
 The variable $s$ keeps track of the kinematic state of the vehicle 
 that is abstracted by the control states of the automaton \textbf{Ac} (accelerating), \textbf{Br} (braking) and \textbf{Csp} (moving with constant speed).

As the updates of $F$ are sporadic, 
 we use a local variable $F'$ to keep track of the possible changes of $F$ since its latest update as follows.
 When the vehicle completes an accelerating or braking step (i.e., when the completion transitions $ca$ or $cb$ occur),
 $F'$ is updated by $F'- A(v_i, v_{i+1})$ and $F'- B(v_i, v_{i-1})$, respectively.
 When the vehicle is moving with constant speed $v$, 
 $F'$ is updated by $F'- v * \triangle t$ for every time period $\triangle t$,
 where $\triangle t$ is a time constant such that the upper bound of the uncertainty $\epsilon = v_{n} * \triangle t$ in the estimation of $F$ remains small.

When an \textbf{UpdateF} occurs, from any state the controller moves to an update state \textbf{Upd}.
 This is represented by grouping the three locations of the automaton into a macro state with an outgoing transition to location \textbf{Upd}.
 Then without delay from location \textbf{Upd}, 
 the controller compares $F'$ to the corresponding bounds and moves to a safe target location accordingly.

As for the synchronous controller, 
 the upper and lower bounds of the triggering conditons of \textbf{bb} and \textbf{ba} are modified so as to take into account the uncertainty $\epsilon$ in the computation of $F$.
 In fact if $x$ is the remaining distance to travel since the last update of $F$, we have $x \geq F' \geq x - \epsilon $. 
 So the condition for accelerating to speed $v_{i+1}$ from $v_i$ becomes $F' \geq D'_{i+1}$,
 where $D'_{i+1} = D_{i+1} + \epsilon $ and the condition for braking from $v_i$ becomes $B''_i \geq F' \geq B'_i$ where $B''_i = B_i + 2*\epsilon$ and $B'_i = B_i + \epsilon$.  
 The safety argument still holds for this asynchronous controller because it simply applies after each update of $F$ an A/B policy for the considered speed levels.

\begin{theorem} \label{theorem4}
Both the synchronous and the asynchronous controller yield safe control policies for collision avoidance.
\end{theorem}

\begin{proof}
The safety proof for the synchronous controller follows the same reasoning as in the previous theorem .

The safety proof for the asynchronous controller is by induction on the updates of the free distance $F$.
 Assume that an update starts and the vehicle is at speed level $v_i$ that is safe for $F'=F$. 
 Then until the next update occurs, the controller will keep track of the free distance by updating $F'$ in every $\triangle t$ time: 
 $F' := F' - v_i * \triangle t$ and in that manner at any time $F'$ will be such that $F' + \epsilon \geq x$,
 where $x$ is the remaining distance to safely travel since the last update of $F$.
 The controller applies an A/B policy, which is safe because the conditions for accelerating and braking have been modified to take into account the maximal deviation $\epsilon$.  
 So, the vehicle will move safely until the next update or stop after $F'$ reaches a value $B''_{1} \geq F' \geq B'_{1}$ and trigger the braking command $Br(v_1, 0)$. 
\end{proof}

Following the same reasoning as in the previous theorem, 
 we can deduce that the efficiency of the two controllers strictly increases for increasing number of speed levels $n$.
 Furthermore, the efficiency depends on how frequently $F$ is updated as the 
 accelerating and braking conditions take into account the uncertainty about the values of $F$.
 Thus, the controllers may not be able to fully utilize actually available free distance.
 While as explained for synchronous controller, 
 the loss in efficiency depends on the  value of $v_{n} * T$,
 for the asynchronous one it depends on the maximal time difference between two successive updates of $F$.

\section{Experimental evaluations}
\label{impl-expr}

We have implemented both the synchronous and the asynchronous controller in the open-source autonomous driving simulator Carla \cite{Dosovitskiy17}.
 In the experiments, we consider scenarios where the controlled vehicle is driving towards a moving vehicle ahead as shown in Fig.\ref{carla}.
 The speed of the front vehicle is described by the periodic function $v_{f}(t) = v_{f0} + v_{f0} * sin(\omega * t)$,
 where $\omega = 2*\pi / T_{f}$, $T_{f}$ is the period of this speed function and $v_{f0}$ is a constant.
 We choose $v_{f0} = 14 ~ m/s$, and thus the speed of the front vehicle changes in the interval $[0, 28 ~ m/s]$ (i.e., $[0, 100.8~ km/h]$).
 We set the limit speed of the controlled vehicle to be $32 ~ m/s$ (i.e., $115.2 ~ km/h$).
 The initial distance between the two vehicles is $F(0) = 5 m$ and the initial speed of the controlled vehicle is 0.
 Thus, the controlled vehicle is initially at a safe state.
 The accelerating and braking rates of the two vehicles are both constant $a=b=2 ~ m/s^{2}$.

\begin{figure}[h]
\centering
\includegraphics[scale = 0.2]{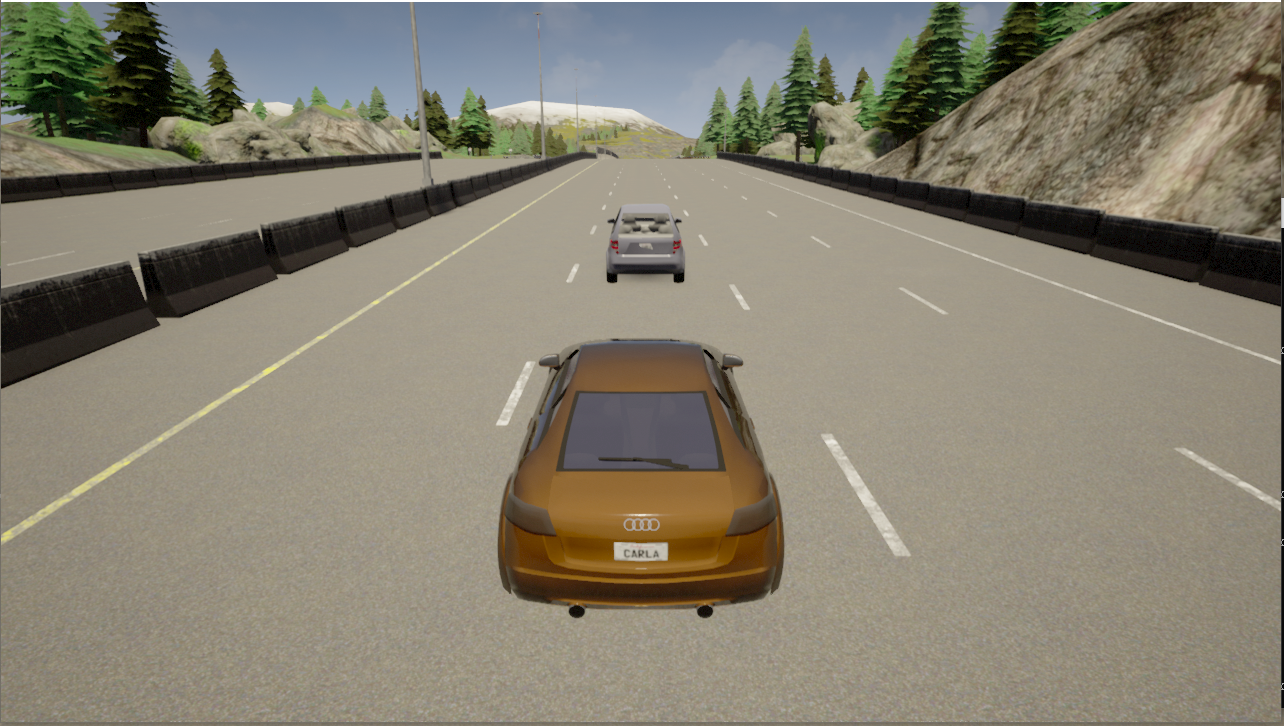}
\caption{Simulation environment in Carla}
\label{carla}
\end{figure}

In order to evaluate the performance and the quality of the controllers,
 we measure both the speed changes of the controlled vehicle 
 and the relative distance between the two vehicles, reflecting the occupancy of the road.
 The smaller the distance is, the higher the road occupancy is.
 We perform experimental evaluations in two settings.

 \begin{itemize}
   \item In \textbf{Setting 1}, 
   we consider that the free distance ahead is equal to the relative distance between the two vehicles, that is we ignore the speed of the front vehicle. 
   This corresponds to a strict safety policy that avoids collision even when the front vehicle suddenly stops e.g. in case of accident.
   \item In \textbf{Setting 2}, 
   we consider that the free distance is the relative distance increased by the braking distance of the front vehicle.
 \end{itemize}

In both settings,  we perform the evaluations with respect to three parameters:
 the period $T_{f}$ of $v_{f}$, the number of speed levels $n$ of the controlled vehicle and the period $T$ of sensing the free distance.
 For experimental purposes,
 the safe accelerating and braking distances for eight speed levels $v[8] = \{4, 8, 12, 16, 20, 24, 28, 32\}$
 are pre-computed as shown in Table.\ref{tab-distances} and configured in the implementations.
 The distances are obtained for constant accelerating and braking rates $a = b =2 ~ m/s^{2}$.

\begin{table}
\centering
\caption{Safe accelerating and braking distances for the eight speed levels}
\begin{tabular}{|p{2cm}|p{2.5cm}|p{2.5cm}|p{2.5cm}|}
  \hline
  speed level ($m/s$)  & Accelerating distance ($m$) & Braking distance ($m$) & Distance for A/B policy  \\ \hline
  $v_{1}$ = 4   & $A(v_0, v_1)$ = 4    & $B(v_{1})$ =  4  & $D(v_0, v_1)$ = 8   \\ \hline
  $v_{2}$ = 8   & $A(v_1, v_2)$ = 12   & $B(v_{2})$ = 16  & $D(v_1, v_2)$ = 28  \\ \hline
  $v_{3}$ = 12  & $A(v_2, v_3)$ = 20   & $B(v_{3})$ = 36  & $D(v_2, v_3)$ = 56  \\ \hline
  $v_{4}$ = 16  & $A(v_3, v_4)$ = 28   & $B(v_{4})$ = 64  & $D(v_3, v_4)$ = 92  \\ \hline
  $v_{5}$ = 20  & $A(v_4, v_5)$ = 36  & $B(v_{5})$ = 100 & $D(v_4, v_5)$ = 136  \\ \hline
  $v_{6}$ = 24  & $A(v_5, v_6)$ = 44  & $B(v_{6})$ = 144 & $D(v_5, v_6)$ = 188  \\ \hline
  $v_{7}$ = 28  & $A(v_6, v_7)$ = 52  & $B(v_{7})$ = 196 & $D(v_6, v_7)$ = 248  \\ \hline
  $v_{8}$ = 32  & $A(v_7, v_8)$ = 60  & $B(v_{8})$ = 256 & $D(v_7, v_8)$ = 316  \\ \hline
\end{tabular}
\label{tab-distances}
\end{table}

\begin{figure}[htbp]
\centering
\begin{minipage}[t]{0.48\textwidth}
\centering
\includegraphics[width = \textwidth]{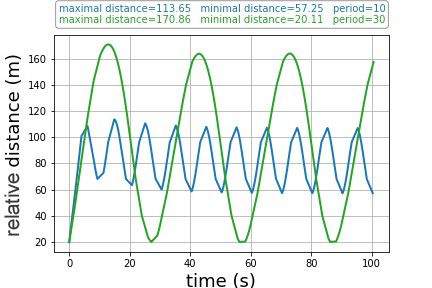}
\end{minipage}
\begin{minipage}[t]{0.48\textwidth}
\centering
\includegraphics[width = \textwidth]{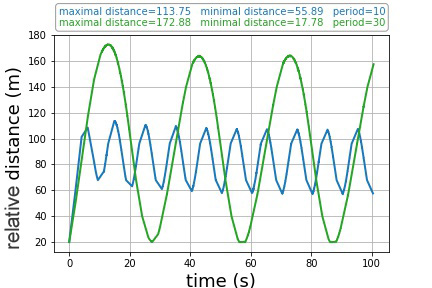}
\end{minipage}
\begin{minipage}[t]{0.48\textwidth}
\centering
\includegraphics[width = \textwidth]{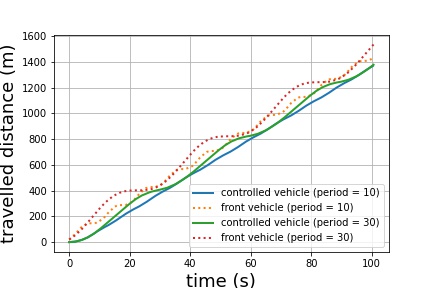}
\end{minipage}
\begin{minipage}[t]{0.48\textwidth}
\centering
\includegraphics[width = \textwidth]{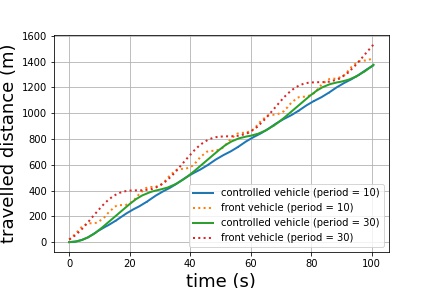}
\end{minipage}
\begin{minipage}[t]{0.48\textwidth}
\centering
\includegraphics[width = \textwidth]{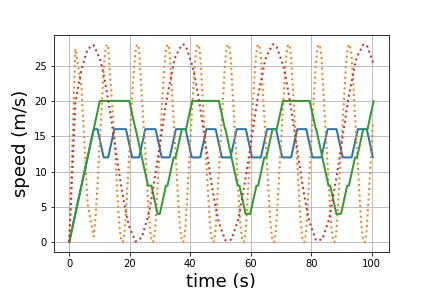}
\end{minipage}
\begin{minipage}[t]{0.48\textwidth}
\centering
\includegraphics[width = \textwidth]{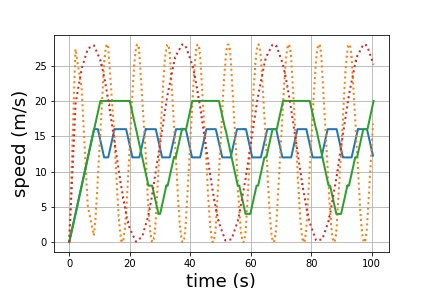}
\end{minipage}
\caption{Simulation results for the synchronous (left part) and the asynchronous controller (right part) 
 for $T_f \in \{10 ~ s, 30 ~s\}$ and  $n=8$ speed levels and sensing period $T=0.02~s$.}
\label{sync-res-1}
\end{figure}

\begin{figure}
\centering
\begin{minipage}[t]{0.49\textwidth}
\centering
\includegraphics[width = \textwidth]{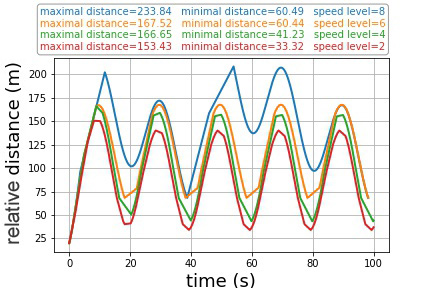}
\end{minipage}
\begin{minipage}[t]{0.49\textwidth}
\centering
\includegraphics[width = \textwidth]{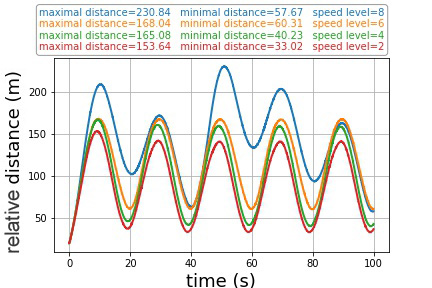}
\end{minipage}
\begin{minipage}[t]{0.49\textwidth}
\centering
\includegraphics[width = \textwidth]{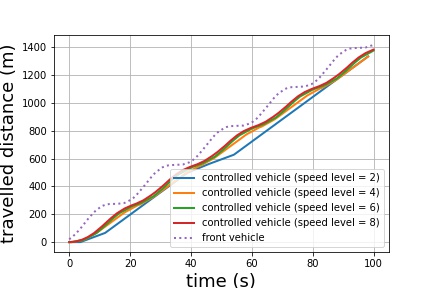}
\end{minipage}
\begin{minipage}[t]{0.49\textwidth}
\centering
\includegraphics[width = \textwidth]{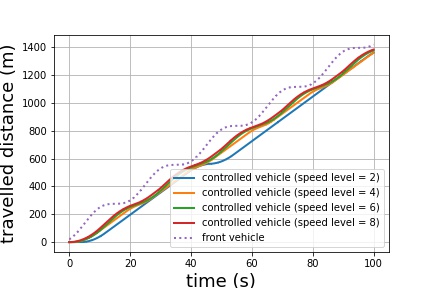}
\end{minipage}
\begin{minipage}[t]{0.49\textwidth}
\centering
\includegraphics[width = \textwidth]{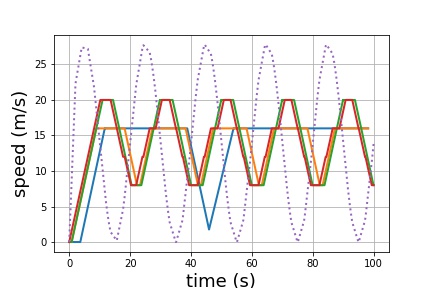}
\end{minipage}
\begin{minipage}[t]{0.49\textwidth}
\centering
\includegraphics[width = \textwidth]{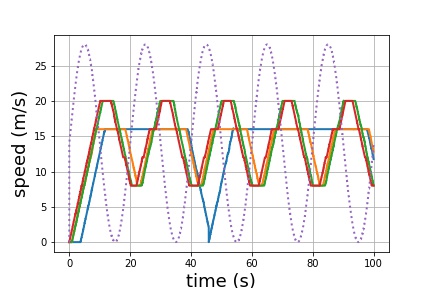}
\end{minipage}
 \caption{Simulation results for the synchronous (left part) and the asynchronous (right part) controller
 for four different speed levels ($n \in \{2,4,6,8\}$) and $T_{f} = 20 ~ s$, $T = 0.02~s$. }
\label{sync-res-2}
\end{figure}

\begin{figure}
\centering
\begin{minipage}[t]{0.49\textwidth}
\centering
\includegraphics[width = \textwidth]{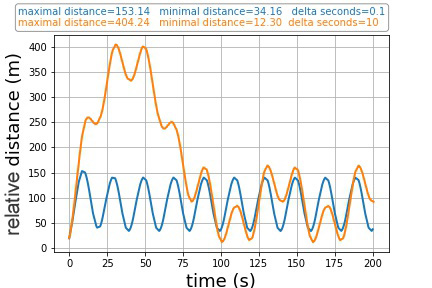}
\end{minipage}
\begin{minipage}[t]{0.49\textwidth}
\centering
\includegraphics[width = \textwidth]{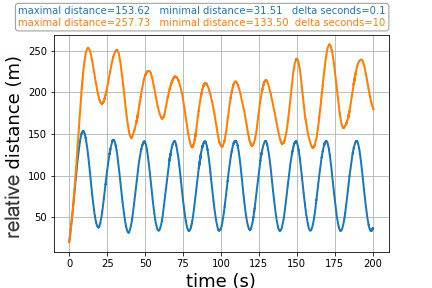}
\end{minipage}
\begin{minipage}[t]{0.49\textwidth}
\centering
\includegraphics[width = \textwidth]{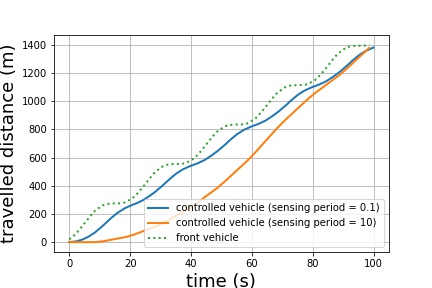}
\end{minipage}
\begin{minipage}[t]{0.49\textwidth}
\centering
\includegraphics[width = \textwidth]{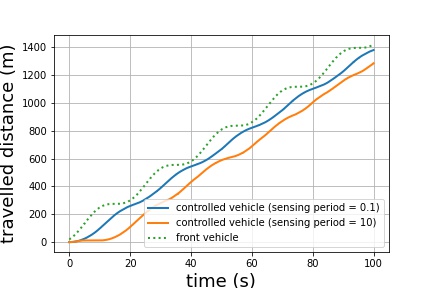}
\end{minipage}
\begin{minipage}[t]{0.49\textwidth}
\centering
\includegraphics[width = \textwidth]{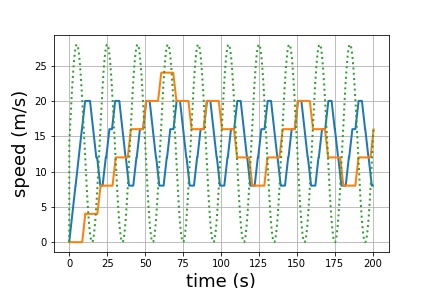}
\end{minipage}
\begin{minipage}[t]{0.49\textwidth}
\centering
\includegraphics[width = \textwidth]{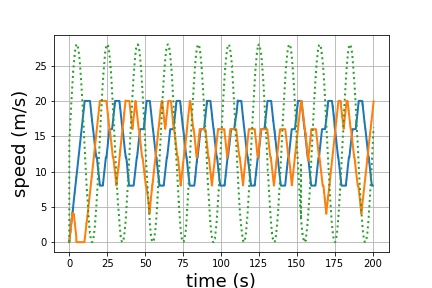}
\end{minipage}
\caption{Simulation results for the synchronous (left part) and the asynchronous (right part) controller
 for two different sensing periods ($T \in \{0.1~s, 10~s\}$) and $T_{f} = 20 ~ s$, $n=8$. }
\label{sync-res-3}
\end{figure}

\textbf{Setting 1}: First, we evaluate how $T_{f}$ affects the performance for two different values $T_{f} = 10 ~ s$ and $T_{f} = 30 ~ s$.
 We assume that the environment updates the free distance variable with period $T = 0.02~s$,
 and that the asynchronous controller estimates this variable with period $\triangle t = 0.005~s$
 (as shown in the controller in Fig.\ref{controller2}).
 
 Fig.\ref{sync-res-1} compares the simulation results for the synchronous (left part) and the asynchronous controller (right part).
 The top two figures compare the dynamics of the free distances.
 For both controllers, the relative distance is periodic with period $T_{f}$ in the steady regime.
 It decreases for increasing period $T_{f}$.
 For instance, it is around $57.27~m$ for $T_{f} = 10 ~s$, while for $T_{f} = 30 ~s$ it is $20.11~m$.
 In fact, for slower speed changes,
 the controller has more time to adjust the movement of the controlled vehicle and can better utilize the available distance.

Note that the results are similar for the two controllers when period $T_{f}$ changes.
 A minor difference is that the minimal relative distances are smaller for the asynchronous controller: when $T_{f} = 30~s$,
 the minimal relative distance of $20.11~m$ for the synchronous controller reduces to $17.78~m$ for the asynchronous controller.
 The bottom figures compare the speed changes of the controlled vehicle (solid lines)
 in response to the speed changes of the front vehicle (dotted lines).
 Similarly, the speed of the controlled vehicle gets closer to $v_f$ for increasing period $T_{f}$.
 For both controllers, the maximal speeds increase from $16 ~ m/s$ to $20 ~m/s$ when the period increases from $10~s$ to $30~s$.

In the second set of experiments,
 we evaluate how the number of speed levels $n$ affects the performance of the two controllers
 for four different values of $n=2$, $n=4$, $n=6$ and $n=8$.
 The period of the front vehicle's speed function is $T_{f} = 20 ~ s$ and the sensing period is $T = 0.02~s$.

The results for the synchronous controller are shown in the left part of Fig.\ref{sync-res-2}.
 We can see that for decreasing number of speed levels, the relative distance increases.
 For instance, when $n=8$, the minimal relative distance is $33.32 ~ m$, which becomes $60.49$ when $n=2$.
 Thus, occupancy deteriorates when less speed levels are used.
 This result simply confirms a consequence of Theorem 3.
 The bottom figure compares the speed changes of the controlled vehicle in response to the speed changes of the front vehicle.
 The result similarly shows that the maximal speed of the controlled vehicle increases from $16~ m/s$ to $20~m/s$
 with the number of speed levels.
 However, beyond a certain number, the performance improvement is negligible:
 e.g., the speed curves for $n=6$ and $n=8$ are almost identical.

For the asynchronous controller, similar results are shown in the right part of Fig.\ref{sync-res-2}.
 For decreasing number of speed levels, the relative distance increases.
 The minimal relative distance changes from $33.02$ to $57.61~m$ when the speed level reduces from 8 to 2.
 These benchmarks show that varying the number of speed levels
 does not result in significant performance differences between synchronous and asynchronous controller.
 However, the relative distances for the considered speed levels are slightly smaller for the asynchronous controller.

In the third set of experiments, we evaluate how the period $T$ of sensing the free distance ahead affects the performance
 for  $T = 0.1~s$ and $T = 10~s$ with $n=8$ and $T_{f} = 20 ~ s$.
The results for the synchronous controller are shown in the left part of Fig.\ref{sync-res-3}.
 Note the transient behavior when the period $T$ increases.
 In the speed diagram the controlled vehicle accelerates from 0 to the highest speed $24~m/s$
 because the largest relative distance is reached only in the transient phase.
 Furthermore, for increasing period, the range of the relative distance increases.
 When $T = 0.1~s$, the relative distance changes in the interval $[34.16, 153.14]~m$ , which becomes $[12.30, 161.60]~m$ when $T = 10~s$
 as the uncertainty about the free distance increases with the period.

For the asynchronous controller, the results are shown in the right part of Fig.\ref{sync-res-3}.
 Note that period of sensing has significant impact on the performance of the asynchronous controller.
 For increasing period,  the relative distance considerably increases.
 Furthermore, compared to the synchronous controller, there is no obvious oscillation during the transient phase.
 This is because the asynchronous controller computes an estimation of the free distance at each local time step.
 As a result, the controlled vehicle will not accelerate to the highest speed $24~m/s$, which is the possible for the synchronous controller.

\begin{figure}[htbp]
\centering
\begin{minipage}[t]{0.48\textwidth}
\centering
\includegraphics[width = \textwidth]{group1-sync_dis.jpg}
\end{minipage}
\begin{minipage}[t]{0.48\textwidth}
\centering
\includegraphics[width = \textwidth]{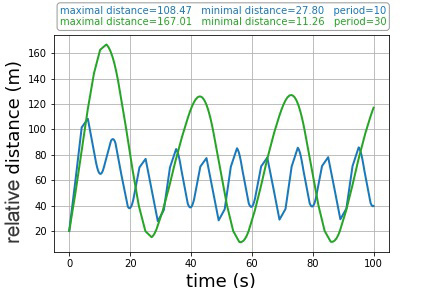}
\end{minipage}
\begin{minipage}[t]{0.48\textwidth}
\centering
\includegraphics[width = \textwidth]{group1-sync_pos.jpg}
\end{minipage}
\begin{minipage}[t]{0.48\textwidth}
\centering
\includegraphics[width = \textwidth]{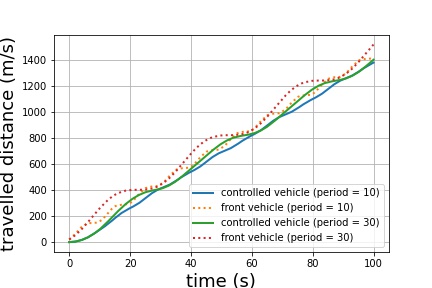}
\end{minipage}
\begin{minipage}[t]{0.48\textwidth}
\centering
\includegraphics[width = \textwidth]{group1-sync_vel.jpg}
\end{minipage}
\begin{minipage}[t]{0.48\textwidth}
\centering
\includegraphics[width = \textwidth]{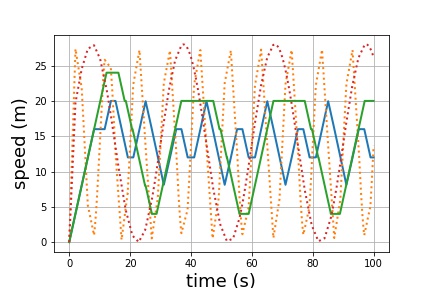}
\end{minipage}
\caption{Simulation results for the synchronous controller with and without considering the braking distance of the front vehicle (right and left part, respectively) for $T_f \in \{10~s, 30~s\}$. }
\label{res-4}
\end{figure}

\begin{figure}
\centering
\begin{minipage}[t]{0.49\textwidth}
\centering
\includegraphics[width = \textwidth]{group2-sync_dis.jpg}
\end{minipage}
\begin{minipage}[t]{0.49\textwidth}
\centering
\includegraphics[width = \textwidth]{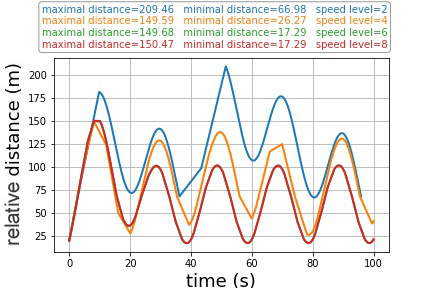}
\end{minipage}
\begin{minipage}[t]{0.49\textwidth}
\centering
\includegraphics[width = \textwidth]{group2-sync_pos.jpg}
\end{minipage}
\begin{minipage}[t]{0.49\textwidth}
\centering
\includegraphics[width = \textwidth]{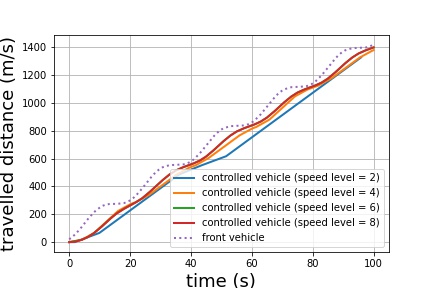}
\end{minipage}
\begin{minipage}[t]{0.49\textwidth}
\centering
\includegraphics[width = \textwidth]{group2-sync_vel.jpg}
\end{minipage}
\begin{minipage}[t]{0.49\textwidth}
\centering
\includegraphics[width = \textwidth]{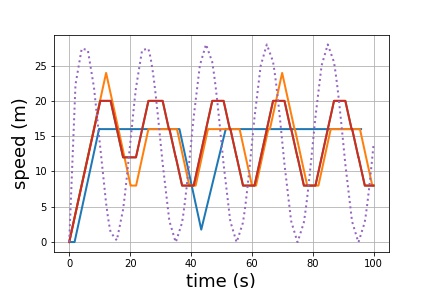}
\end{minipage}
\caption{Simulation results for the synchronous controller
 with and without considering the braking distance of the front vehicle (right and left part, respectively)
 for four different speed levels $n \in \{2, 4, 6, 8\}$.}
 \label{res-5}
\end{figure}

\begin{figure}
\centering
\begin{minipage}[t]{0.49\textwidth}
\centering
\includegraphics[width = \textwidth]{group3-sync_dis.jpg}
\end{minipage}
\begin{minipage}[t]{0.49\textwidth}
\centering
\includegraphics[width = \textwidth]{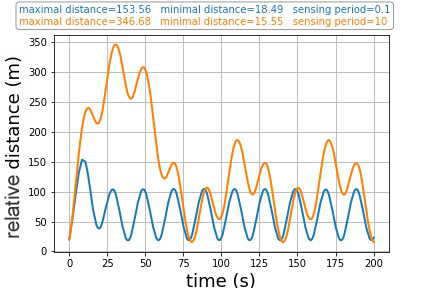}
\end{minipage}
\begin{minipage}[t]{0.49\textwidth}
\centering
\includegraphics[width = \textwidth]{group3-sync_pos.jpg}
\end{minipage}
\begin{minipage}[t]{0.49\textwidth}
\centering
\includegraphics[width = \textwidth]{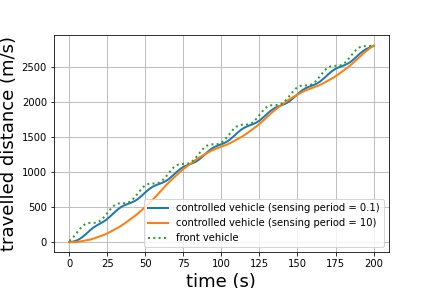}
\end{minipage}
\begin{minipage}[t]{0.49\textwidth}
\centering
\includegraphics[width = \textwidth]{group3-sync_vel.jpg}
\end{minipage}
\begin{minipage}[t]{0.49\textwidth}
\centering
\includegraphics[width = \textwidth]{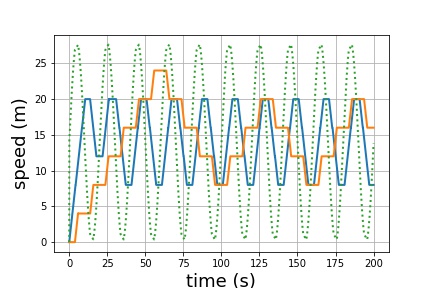}
\end{minipage}
\caption{Simulation results for the synchronous controller
 with and without considering the braking distance of the front vehicle (right and left part, respectively) for two sensing periods $T \in \{0.1~s, 10~s\}$.}
\label{res-6}
\end{figure}

\textbf{Setting 2}:
 in the subsequent experiments, we evaluate how performance changes 
 when we take into account in the evaluation of the free distance the braking distance of the front vehicle travelling with speed $v_{f}(t) = v_{f0} + v_{f0} * sin(\omega * t)$.
 At time $t$ its braking distance is $v_{f}(t)^{2}/2*b_{f}$ for a constant braking rate $b_{f}$.
 We take $b_{f} = 5 m/s^{2}$ for experimental purposes.
 Then the corresponding safe accelerating and braking distances are $D'_{i} = D_{i} - v_{f}(t)^{2}/2*b_{f}$ and $B'_{i} = B_{i} - v_{f}(t)^{2}/2*b_{f}$.

 In Fig.\ref{res-4}, Fig.\ref{res-5} and Fig.\ref{res-6},
 we compare the performance of the synchronous controller
 with and without considering the braking distance of the front vehicle (right and left part, respectively).
 In Fig.\ref{res-4} we compare the results for $T_f \in \{10~s, 30~s\}$ with sensing period $T = 0.02 ~ s$ and speed level $n=8$.
 We can see that when taking into account the braking distance of the front vehicle, the relative distance between the two vehicles becomes much smaller.
 For instance, the minimal relative distance decreases from $20.11 ~ m$ for $T_{f} = 30 ~s$ to $11.26 ~ m$.
 Furthermore, the distance decreases when the period $T_{f}$ increases as shown in the top figures in Fig.\ref{res-4}.
 The performance improvement can also be observed from the speed diagrams shown in the bottom of Fig.\ref{res-4}.
 Finally, the speed range of the controlled vehicle becomes larger and
 the maximal speed increases from $16~m/s$  to $20~m/s$ for $T_f = 10~s$.

Similar results are shown in Fig.\ref{res-5} and Fig.\ref{res-6},
 which compare the results for four different numbers of speed levels  $n \in \{2, 4, 6, 8\}$ 
 and two different sensing periods $T \in \{0.1~s, 10~s\}$, respectively.
 The relative distance reduces significantly when taking into account the braking distance of the front vehicle.
 For instance,  the minimal free distance drops from $33.32~m$ to $17.29~m$ when $n=8$. 
 These results confirm that by taking into account the movements of the front vehicle,
 the controlled vehicle can move more aggressively and better utilize the free distance ahead.

\section{Conclusions and future work}
\label{conclusions}

The paper presents a novel framework and approach for safe and efficient collision avoidance for self-driving vehicles.
 The framework is model-based and assumes that control policies are implemented as the application of acceleration, braking and constant speed commands.
 The presented algorithms do not make any assumption about the dynamics of the controlled vehicle except
 that there are functions giving the traveled distance when the speed of the vehicle changes by some quantity.
 Additionally, the assumptions about the vehicle's environment are minimal as it is described by a function $F(t)$ giving at any time the free available distance ahead.
 
The assumption that all vehicles move in the same direction as the controlled vehicle does not limit the generality of our approach.
 The same algorithm can be applied by adequately taking into account movements in the different directions in the estimation of the free headway distance.
 If for instance the distance between the controlled vehicle and the front vehicle moving in the opposite direction is $F'$ at time $t$,
 then the free space ahead can be estimated as $F = F' - D(\triangle t)$,
 where $\triangle t$ is the time needed for the controlled vehicle to completely stop by braking from its current speed,
 and $D(\triangle t)$ is the maximal distance travelled by the front vehicle within time $\triangle t$.
 
The same algorithm can be adapted to two-dimensional movement. 
 In that case, the function $F(t)$ can be defined as the maximal convex area containing the controlled vehicle and such that all the obstacles are outside this area.  
 The distances $A(V,v)$ and $B(V,v)$ for initial and target speeds respectively are also replaced by adequately approximated convex areas 
 so that the safety test boils down to area inclusion that can be efficiently decided.

The presented approach differs from others based on control theory or controller synthesis in that it guarantees at a high level of abstraction both safety and efficiency.
 We progressively relax the assumption about perfect real-time knowledge of the free space and provide solutions that are safe and relatively efficient even
 when the free space is sporadically updated.
 Furthermore, switching between a set of speed levels depending on pre-computed conditions drastically reduces the computational complexity of the decision process.
 This also allows to reduce the control algorithm sensitivity to changes of the environment while keeping driving safe and robust.
 Experimental results show that it is easy to get implementations under minimal assumptions about the operational environment.

This work is part of a larger project on the design of safe and efficient autopilots for self driving cars.
 Future developments include the adaptation of this algorithm to two-dimension movement
 where the free space function provides areas around the controlled vehicle as well as  the integration of our algorithms in autonomous car models of the Carla simulator.


\bibliographystyle{splncs03}
\bibliography{main}

\end{document}